\RequirePackage{amsmath}
\documentclass{llncs}
\newcommand{\bfs}{\mathrm{BFS}}
\newcommand{\notinbfs}{black!25}

\usepackage{amssymb}
\usepackage{subfigure}
\usepackage{tikz}
\usetikzlibrary{decorations.pathreplacing}
\usetikzlibrary{backgrounds}

\begin{document}

\title{BFS Enumeration for Breaking Symmetries in Graphs}
%
%
\author{Vyacheslav Moklev\inst{1, 2} \and Vladimir Ulyantsev\inst{1}}

\institute{ITMO University, Saint-Petersburg, Russia \and 
JetBrains Research, Saint-Petersburg, Russia\\
\email{\{moklev, ulyantsev\}@rain.ifmo.ru}
}

\maketitle              

\begin{abstract}
There are numerous NP-hard combinatorial problems which involve searching for an undirected graph satisfying a certain property. One way to solve such problems is to translate a problem into an instance of the boolean satisfiability (SAT) or constraint satisfaction (CSP) problem. Such reduction usually can give rise to numerous isomorphic representations of the same graph. One way to reduce the search space and speed up the search under these conditions is to introduce symmetry-breaking predicates. In this paper we introduce three novel and practically effective symmetry-breaking predicates for an undirected connected graph search based on breadth-first search (BFS) enumeration and compare with existing symmetry-breaking methods on several graph problems.
\end{abstract}
\section{Introduction}
\label{sec:1}
The search problems of a certain automaton or a graph are encountered in 
grammatical inference and natural language processing. For most of them it is proved
to be NP-hard or no polynomial solution known. On the other hand, in recent years modern SAT-solvers have been developed
and now they are powerful tools for solving huge SAT instances.
Every year SAT competitions are held and some winners are able to solve SAT instances
with millions of clauses and variables.
For some generalizations of SAT such as CSP and SMT efficient solvers
\cite{nethercote2007minizinc,stuckey2014minizinc} and optimizing compilers into SAT~\cite{metodi2012compiling} exist.
Some of the automaton and graph search problems can be efficiently translated into a SAT instance
and solved by a SAT- or CSP-solver.
In this paper we consider only graph problems, but such translations 
are widely used in many other problems, like identifying matrices~\cite{lynce2007breaking}, scheduling~\cite{crawford1994experiment} etc. 
\par 
In most graph search problems we want to find an unlabeled graph (without enumeration of nodes) and we constrain only a graph's structure, but not it's enumeration. But the most common graph representations force to enumerate all nodes. Such enumeration give rise to numerous representations of the same unlabeled graph, extending the search space and slowing down the search. Such representations are usually called \emph{symmetries} and they usually occur in many problems~\cite{devriendt2016improved,gent1999symmetry}. A common 
technique under these conditions is to use \emph{symmetry-breaking} predicates.
The main idea of symmetry breaking is to introduce some additional constraints
to reduce the number of isomorphic solutions but keep at least one solution from each
isomorphism class. Such constraints allow the SAT solver to find conflicts earlier and thus increase the performance of the search.
\par
There are also some powerful tools (like \texttt{nauty}) that can find all 
isomorphic graphs from a given set. But they are not applicable in our context
since we want to eliminate isomorphs during the search, cut some branches 
of the search tree and reduce the search space.
\par 
Symmetry-breaking constraints for graph problems have been widely studied for the last years.
There are several approaches to construct such constraints.
One popular way to break symmetries in graphs is closely
connected with \emph{canonical} representation of graph~--
lexicographically minimum graph with respect to a certain order
over adjacency matrices. Works~\cite{heule30quest,itzhakov2016breaking} are related to 
finding a ``perfect'' (that eliminates \emph{all} symmetries) symmetry-breaking predicate 
for small graphs (\cite{itzhakov2016breaking} is based on canonicity) and in~\cite{Codish1,miller2012diamond} several properties were found 
which hold for every canonical representation of graph.
\par
A different approach was considered in~\cite{cuong2016computing}. This work is related
to the search of maximum (in the number of edges) unavoidable subgraphs of a 
given complete graph. Enumerating vertices of such subgraph is
a way to break some symmetries. 
\par
A new method of symmetry breaking that eliminates all symmetries in automata
search problems was introduced recently~\cite{ulyantsev2015bfs}. This method
is based on BFS enumeration of an automaton, which is unique for
every isomorphism class. In this work we propose an adaptation of
this approach for undirected graph search problems.
\par
We introduce three symmetry-breaking predicates for an undirected graph search.
The first one is based on the approach from~\cite{ulyantsev2015bfs}. The second and
third ones are improvements of this predicate aiming to eliminate more symmetries.
We prove the correctness of these predicates and compare them with existing ones.
We applied these methods to two combinatorial problems from 
extremal graph theory and conducted an experiment. We implemented
methods from~\cite{Codish1,cuong2016computing} and our best method works faster for almost 
all test cases.

\section{Definitions}
\label{sec:2}

One of the most important notions related to the symmetry-breaking techniques is the \emph{isomorphism} of the different objects. It is typical that a lot of objects' representations are ambiguous (with respect to the problem), i.e. for one
object there are several representations. For each unlabeled graph there are numerous representations that differ only in enumeration of vertices. Such graphs are called \emph{isomorphic}.
Graph search problems are typically invariant under graph isomorphism: for each
isomorphism class either all or none of graphs from this class are solutions.
\par 
During the search a solver has to check several isomorphic graphs, but it is enough to check
only one representative from each isomorphism class. One way to help the solver to avoid the
checking of such symmetrical solutions (therefore speeding up the search) is to introduce
\emph{symmetry-breaking predicates}. 
A symmetry-breaking predicate (SBP) over graphs is a boolean function (constraint) for a graph
that allows at least one graph from each isomorphism class (but as few as possible).
SBP must allow at least one graph from each isomorphism class to prevent 
the loss of solutions. But sometimes we have prior information that 
all solutions of the problem have some property (e.g. all solutions
are connected graphs). In such situations we can eliminate a whole 
isomorphism class of not connected graphs without the loss of solutions
and we have to allow at least one graph only from isomorphism class
of connected graphs. Such predicates are called \emph{instance dependent}~\cite{itzhakov2016breaking}
symmetry-breaking predicates.
\par 
In this work we introduce three instance dependent symmetry-breaking 
predicates for an undirected connected graph search.

\section{Symmetry breaking}
\label{sec:3}

Our approach to break symmetries is based on the idea of \emph{BFS-enumeration}
introduced in~\cite{ulyantsev2015bfs}.

\begin{definition}
Graph $G$ is \emph{BFS-enumerated} $(P_\bfs(G) = 1)$ if there exists a BFS traversal such that for all $k$ from $1$ to $|V(G)|$,
$k$-th vertex in this traversal has the number (label) $k$.
Otherwise ${P_\bfs(G) = 0}$.
\end{definition}

Some examples of BFS-enumerated graphs are shown in 
Figure~\ref{fig:isomorphic2}.
\par

To encode this constraint into CSP we define integer variables 
$p_i,\ i \in 2..|V|$ which denote a label of the parent of node $i$.
Then we constrain these variables like in \cite{ulyantsev2015bfs}
(where $A[i, j]$ is a $(i, j)$-th element of adjacency matrix):
\begin{equation}
\forall i: p_i \le p_{i+1},
\end{equation}
\begin{equation}
\forall i, j: p_j = i \Leftrightarrow A[i, j] \wedge \nexists k < i: A[k, j].
\end{equation}

\par
This is just a translation of an automaton predicate from~\cite{ulyantsev2015bfs} to graph problems.
Actually, $P_\bfs(G)$ is a symmetry-breaking predicate among connected 
graphs (instance dependent SBP)  that we proved in the Theorem \ref{thm:bfs} in Appendix~\ref{apd:A}.
\par
The main drawback of this predicate is that the start vertex could be arbitrary.
In Figure~\ref{fig:isomorphic2} three isomorphic graphs (with BFS-trees) are shown
that are allowed by $P_{\bfs}$ (arrows are arcs from BFS traversal).

One way to fix a start vertex is to choose a vertex with a maximum \emph{degree} (number of adjacent vertices).
\begin{definition}
Let $G$ be a connected graph. Then $P_\bfs^+(G) = 1$ if and only if $P_\bfs(G) = 1$ and 
$\deg{v_1} = \max\limits_{1 \le k \le n}\deg{v_k}$.
\end{definition}

A new predicate $P_\bfs^+(G)$ is also a symmetry-breaking predicate for
connected graphs that we proved in the Theorem \ref{thm:bfs_plus} in Appendix~\ref{apd:A}.

\newcommand{\scalepr}{0.8}

\begin{figure}[ht]
\begin{minipage}[t]{.65\textwidth}
\centering 


\begin{subfigure}
  \centering
    \begin{tikzpicture}
      \tikzstyle{every node}=[draw,shape=circle];
      \node (v1) at (\scalepr * 1, 2 * \scalepr) {\small $1$};
      \node (v2) at (\scalepr * 0, 1 * \scalepr) {\small $4$};
      \node (v3) at (\scalepr * 2, 1 * \scalepr) {\small $2$};
      \node (v4) at (\scalepr * 1, 0 * \scalepr) {\small $3$};
      \tikzset{every node/.style={}}
      \node (d1) at (\scalepr * 1.5, 2 * \scalepr) {\small $\ 1$};
      \node (d2) at (\scalepr * 0.5, 1 * \scalepr) {\small $\ 1$};
      \node (d3) at (\scalepr * 2.5, 1 * \scalepr) {\small $\ 2$};
      \node (d4) at (\scalepr * 1.5, 0 * \scalepr) {\small $\ 2$};
      \draw [arrows={-latex}] (v1) -- (v3);
      \draw [arrows={-latex}] (v3) -- (v4);
      \draw [arrows={-latex}] (v4) -- (v2);
    \end{tikzpicture}
\end{subfigure}%
\hfill
\begin{subfigure}
  \centering
    \begin{tikzpicture}
      \tikzstyle{every node}=[draw,shape=circle];
      \node (v1) at (\scalepr * 1, 2 * \scalepr) {\small $1$};
      \node (v2) at (\scalepr * 0, 1 * \scalepr) {\small $2$};
      \node (v3) at (\scalepr * 2, 1 * \scalepr) {\small $3$};
      \node (v4) at (\scalepr * 1, 0 * \scalepr) {\small $4$};
      \tikzset{every node/.style={}}
      \node (d1) at (\scalepr * 1.5, 2 * \scalepr) {\small $\ 2$};
      \node (d2) at (\scalepr * 0.5, 1 * \scalepr) {\small $\ 2$};
      \node (d3) at (\scalepr * 2.5, 1 * \scalepr) {\small $\ 1$};
      \node (d4) at (\scalepr * 1.5, 0 * \scalepr) {\small $\ 1$};
      \draw [arrows={-latex}] (v1) -- (v2);
      \draw [arrows={-latex}] (v1) -- (v3);
      \draw [arrows={-latex}] (v2) -- (v4);
    \end{tikzpicture}
\end{subfigure}%
\hfill
\begin{subfigure}
  \centering
    \begin{tikzpicture}
      \tikzstyle{every node}=[draw,shape=circle];
      \node (v1) at (\scalepr * 1, 2 * \scalepr) {\small $1$};
      \node (v2) at (\scalepr * 0, 1 * \scalepr) {\small $2$};
      \node (v3) at (\scalepr * 2, 1 * \scalepr) {\small $3$};
      \node (v4) at (\scalepr * 1, 0 * \scalepr) {\small $4$};
      \tikzset{every node/.style={}}
      \node (d1) at (\scalepr * 1.5, 2 * \scalepr) {\small $\ 2$};
      \node (d2) at (\scalepr * 0.5, 1 * \scalepr) {\small $\ 1$};
      \node (d3) at (\scalepr * 2.5, 1 * \scalepr) {\small $\ 2$};
      \node (d4) at (\scalepr * 1.5, 0 * \scalepr) {\small $\ 1$};
      \draw [arrows={-latex}] (v1) -- (v2);
      \draw [arrows={-latex}] (v1) -- (v3);
      \draw [arrows={-latex}] (v3) -- (v4);
    \end{tikzpicture}
\end{subfigure}
\caption{Isomorphic graphs that satisfy $P_\bfs$ predicate. Each vertex is marked with 
its number (inside a circle) and its degree (outside).}
\label{fig:isomorphic2}


\end{minipage}%
\hfill
\begin{minipage}[t]{.31\textwidth}
\centering 

\begin{subfigure}
  \centering
    \begin{tikzpicture}
      \tikzstyle{every node}=[draw,shape=circle];
      \node (v1) at (\scalepr * 1, 2 * \scalepr) {\small $1$};
      \node (v2) at (\scalepr * 0, 1 * \scalepr) {\small $2$};
      \node (v3) at (\scalepr * 2, 1 * \scalepr) {\small $3$};
      \node (v4) at (\scalepr * 1, 0 * \scalepr) {\small $4$};
      \tikzset{every node/.style={}}
      \node (w1) at (\scalepr * 1.7, 1.96 * \scalepr) {\small $\ \ [2, 1]$};
      \node (w2) at (\scalepr * 0.55, 0.96 * \scalepr) {\small $\ \ [1]$};
      \draw [arrows={-latex}] (v1) -- (v2);
      \draw [arrows={-latex}] (v1) -- (v3);
      \draw [arrows={-latex}] (v2) -- (v4);
      \draw [\notinbfs] (v3) -- (v4);
    \end{tikzpicture}
\end{subfigure}
\caption{Counterexample for non-descending weights}
\label{fig:antiexample}

\end{minipage}

\end{figure}

To encode $P_\bfs^+$ into CSP we introduce additional integer variables ${deg_i,\ i \in 1..|V|}$ and $deg_{max}$ which denote degrees of each node in graph and the maximum degree among the nodes in the graph respectively. Then we add the following constraints to define these variables:
\begin{equation}
\forall i: deg_i \le deg_{max},
\end{equation}
\begin{equation}
deg_{max} = \max_{i \in 1..|V|} deg_i,
\end{equation}
\begin{equation}
\forall i: deg_i = \sum_{j = 1}^{|V|} A[i, j],
\end{equation}
\begin{equation}
deg_1 = deg_{max}.
\end{equation}

In Figure~\ref{fig:isomorphic2} three isomorphic graphs which satisfy $P_\bfs$ are shown but the first is not allowed by $P_\bfs^+$. 
So $P_\bfs^+(G)$ eliminates more symmetries than $P_\bfs$, but still not all of them. We partially 
solve an issue of arbitrary start vertex, but we rest a lot of symmetries.
BFS traversal partitions all vertices into \emph{layers}~-- sets of vertices
of equal depth in a BFS-tree (distance from $v_1$). But neither $P_\bfs$ nor $P_\bfs^+$
constrains the order of
vertices in the layer. So after fixing a start vertex (and therefore fixing all layers), the
order of vertices within a layer may be arbitrary.
\par 
To eliminate this kind of symmetry we propose an approach of ordering vertices
in a layer by \emph{weight of subtree}. The weight of subtree for a vertex $v$ is a number of \emph{descendants} of $v$ (vertices below $v$)
in a the BFS-tree, including the vertex $v$ itself. So, the weight of subtree for the start vertex $v_1$
equals the number of vertices in a graph and the weight of subtree for a leaf is equal to~$1$.
An example of different enumerations based on reordering of vertices in the layer
is shown in Figure~\ref{fig:lemma3.1.1} from Appendix~\ref{apd:A}.
\par 
Now we can introduce the symmetry-breaking predicate based on the weight of subtree.

\begin{definition}
\label{def:bfs_star}
Let $G$ be a connected graph. $P_\bfs^*(G) = 1$ if and only if 
$P_\bfs^+(G) = 1$ and for any vertex $v \in V(G)$ children of $v$ in BFS-tree 
are sorted by weight of subtree: 
$\forall v \in V(G): w(\mathrm{child}(v)_1) \ge w(\mathrm{child}(v)_2) \ge \ldots \ge w(\mathrm{child}(v)_k)$, where $k = |\mathrm{child}(v)|$.
\end{definition}

To encode $P_\bfs^*$ in CSP we introduce new integer variables $w_i,\ i \in 1..|V|$ which denote a weight of subtree of this node in the BFS-tree. To define them we add the following constraints ($[a] = 1$ if $a$ is $true$, $[a] = 0$ otherwise):
\begin{equation}
\forall i: p_i = p_{i+1} \Rightarrow w_i \ge w_{i + 1},
\end{equation}
\begin{equation}
\forall i: w_i = 1 + \sum_{j = i + 1}^{|V|} w_j\cdot [p_j = i].
\end{equation}

$P_\bfs^*$ is also a symmetry-breaking predicate, what we proved in Theorem \ref{thm:bfs_star} in Appendix~\ref{apd:A}. An example of an unlabeled graph and 
a proper enumeration allowed by $P_\bfs^*$ in shown in Figure~\ref{fig:sophisticated_example}.
Note that Theorem \ref{thm:bfs_star} is not so trivial as might first appear.
In fact, were we to change the order, taking $w_1 \le w_2 \le \ldots \le w_n$ 
instead, it would not define a symmetry breaking constraint.
The smallest counterexample is shown in Figure~\ref{fig:antiexample}.
There exists only one BFS-numbered graph isomorphic to $C_4$ (cycle of 4 vertices). But $w(1) = 2$ and $w(2) = 1$
so there is no graph isomorphic to $C_4$, BFS-numbered and with
non-descending weights.
\par 

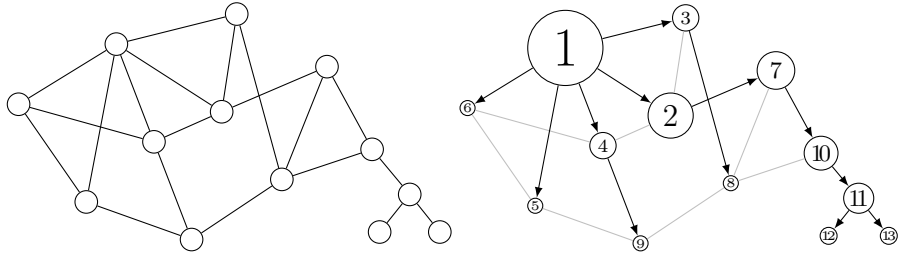
\begin{figure}[ht]
\centering
\begin{subfigure}
  \centering
    \begin{tikzpicture}
      \tikzstyle{every node}=[draw,shape=circle];
      \node[minimum size = 0.3cm, inner sep = 0pt] (v1)  at (1.6, 3.5) {};
      \node[minimum size = 0.3cm, inner sep = 0pt] (v2)  at (3.0, 2.6) {};
      \node[minimum size = 0.3cm, inner sep = 0pt] (v3)  at (3.2, 3.9) {};
      \node[minimum size = 0.3cm, inner sep = 0pt] (v4)  at (2.1, 2.2) {};
      \node[minimum size = 0.3cm, inner sep = 0pt] (v5)  at (1.2, 1.4) {};
      \node[minimum size = 0.3cm, inner sep = 0pt] (v6)  at (0.3, 2.7) {};
      \node[minimum size = 0.3cm, inner sep = 0pt] (v7)  at (4.4, 3.2) {};
      \node[minimum size = 0.3cm, inner sep = 0pt] (v8)  at (3.8, 1.7) {};
      \node[minimum size = 0.3cm, inner sep = 0pt] (v9)  at (2.6, 0.9) {};
      \node[minimum size = 0.3cm, inner sep = 0pt] (v10) at (5.0, 2.1) {};
      \node[minimum size = 0.3cm, inner sep = 0pt] (v11) at (5.5, 1.5) {};
      \node[minimum size = 0.3cm, inner sep = 0pt] (v12) at (5.1, 1.0) {};
      \node[minimum size = 0.3cm, inner sep = 0pt] (v13) at (5.9, 1.0) {};
      
      \draw (v6) -- (v4);
      \draw (v6) -- (v5);
      \draw (v5) -- (v9);
      \draw (v2) -- (v4);
      \draw (v3) -- (v2);
      \draw (v9) -- (v8);
      \draw (v8) -- (v7);
      \draw (v8) -- (v10);
      
      \draw (v1) -- (v2);
      \draw (v1) -- (v3);
      \draw (v1) -- (v4);
      \draw (v1) -- (v5);
      \draw (v1) -- (v6);
      \draw (v2) -- (v7);
      \draw (v3) -- (v8);
      \draw (v4) -- (v9);
      \draw (v7) -- (v10);
      \draw (v10) -- (v11);
      \draw (v11) -- (v12);
      \draw (v11) -- (v13);
      
    \end{tikzpicture}
\end{subfigure}%
\begin{subfigure}
  \centering
    \begin{tikzpicture}
      \tikzstyle{every node}=[draw,shape=circle];
      \node[minimum size = 1.0cm,  inner sep = 0pt] (v1)  at (1.6, 3.5) {\huge $1$};
      \node[minimum size = 0.6cm,  inner sep = 0pt] (v2)  at (3.0, 2.6) {\large $2$};
      \node[minimum size = 0.35cm, inner sep = 0pt] (v3)  at (3.2, 3.9) {\scriptsize $3$};
      \node[minimum size = 0.35cm, inner sep = 0pt] (v4)  at (2.1, 2.2) {\scriptsize $4$};
      \node[minimum size = 0.2cm,  inner sep = 0pt] (v5)  at (1.2, 1.4) {\tiny $5$};
      \node[minimum size = 0.2cm,  inner sep = 0pt] (v6)  at (0.3, 2.7) {\tiny $6$};
      \node[minimum size = 0.5cm,  inner sep = 0pt] (v7)  at (4.4, 3.2) {\small $7$};
      \node[minimum size = 0.2cm,  inner sep = 0pt] (v8)  at (3.8, 1.7) {\tiny $8$};
      \node[minimum size = 0.2cm,  inner sep = 0pt] (v9)  at (2.6, 0.9) {\tiny $9$};
      \node[minimum size = 0.45cm, inner sep = 0pt] (v10) at (5.0, 2.1) {\footnotesize $1\!0$};
      \node[minimum size = 0.4cm,  inner sep = 0pt] (v11) at (5.5, 1.5) {\footnotesize $1\!1$};
      \node[minimum size = 0.2cm,  inner sep = 0pt] (v12) at (5.1, 1.0) {\tiny $1\!2$};
      \node[minimum size = 0.2cm,  inner sep = 0pt] (v13) at (5.9, 1.0) {\tiny $1\!3$};
      
      \draw [\notinbfs] (v6) -- (v4);
      \draw [\notinbfs] (v6) -- (v5);
      \draw [\notinbfs] (v5) -- (v9);
      \draw [\notinbfs] (v2) -- (v4);
      \draw [\notinbfs] (v3) -- (v2);
      \draw [\notinbfs] (v9) -- (v8);
      \draw [\notinbfs] (v8) -- (v7);
      \draw [\notinbfs] (v8) -- (v10);
      
      \draw [arrows={-latex}] (v1) -- (v2);
      \draw [arrows={-latex}] (v1) -- (v3);
      \draw [arrows={-latex}] (v1) -- (v4);
      \draw [arrows={-latex}] (v1) -- (v5);
      \draw [arrows={-latex}] (v1) -- (v6);
      \draw [arrows={-latex}] (v2) -- (v7);
      \draw [arrows={-latex}] (v3) -- (v8);
      \draw [arrows={-latex}] (v4) -- (v9);
      \draw [arrows={-latex}] (v7) -- (v10);
      \draw [arrows={-latex}] (v10) -- (v11);
      \draw [arrows={-latex}] (v11) -- (v12);
      \draw [arrows={-latex}] (v11) -- (v13);
      
    \end{tikzpicture}
\end{subfigure}
\caption{Unlabeled graph (left) and one of its enumerations (right), satisfying a $P_\bfs^*$ predicate. Size of vertex denotes the weight of subtree of this vertex, black arcs represent a BFS-tree.}
\label{fig:sophisticated_example}
\end{figure}

Note that $P_\bfs^*$ allows only one graph from Figure~\ref{fig:isomorphic2}, because the first graph is not allowed by $P_\bfs^+$
and the third graph has ascending weights of subtree for children of
start vertex: $[1, 2]$. Only second graph is BFS-enumerated, has
a start vertex of maximum degree and has all weights sorted in the 
appropriate order.
\par 
Another example is shown in Figure~\ref{fig:weight_example}. After fixing the start vertex and
all layers, we have 6 possible permutations of vertices inside the second layer. Numbers under 
curly braces are the weights of subtrees of vertices from the second layer. Note that only one 
enumeration satisfies $P_\bfs^*$, because other enumerations produce sequences of weights which
are not sorted or sorted in wrong order, so $P_\bfs^*(G_1) = 1$ and $P_\bfs^*(G_k) = 0$ for $k = 2..6$.

\input{weight_example}

\section{Extremal graph problems}
\label{sec:4}
We applied our methods to the following graph problems to show the efficiency of proposed methods and to compare with existing ones.
\begin{definition} 
$ex(n; G_1, \ldots, G_k)$ is the maximum number of edges in a graph with $n$
vertices and without subgraphs isomorphic to $G_1, \ldots, G_k$.
\par
$EX(n; G_1, \ldots, G_k)$ is the set of extremal graphs~-- with no subgraphs isomorphic to $G_1, \ldots, G_k$ and with maximum number of edges.
\end{definition}
\begin{example}
$ex(n; C_3)$ is the maximum number of edges in triangle-free graph with $n$ vertices. 
Well known that $ex(n; C_3) = \lfloor n^2/4 \rfloor$ and extremal graph is
$K_{\lfloor n/2 \rfloor, \lceil n/2 \rceil}$.
\end{example}
For $ex(n; C_3, C_4)$ and $ex(n; C_4)$ asymptotically precise estimations are known:
in~\cite{Erdos1} it is shown that $ex(n; C_3, C_4) = (1/2 + o(1))^{3/2}n^{3/2}$ and
$ex(n; C_4) = (1/2 + o(1))n^{3/2}$ is known from~\cite{Clapham1}.	

\subsection{Determining of $ex(n; C_3, C_4)$}
\label{sec:4.1}

This problem was considered in~\cite{Garnick2,Abajo1,Abajo2,Garnick1,Wang1}. 
Let $n = |V(G)|$, $m = |E(G)|$. We translate a problem into a CSP instance
by a slightly modified model from~\cite{Codish1}:
\begin{equation}
\label{eqn:constraint1}
\forall {(1\le i < j \le n)}: (A[i, j] \equiv A[j, i] \text{ and } A[i, i] \equiv false),
\end{equation}
\begin{equation}
\label{eqn:constraint2}
\forall {i, j, k}: A[i, j] + A[j, k] + A[k, i] < 3,
\end{equation}
\begin{equation}
\label{eqn:constraint3}
\forall {i, j, k, l}: A[i, j] + A[j, k] + A[k, l] + A[l, i] < 4,
\end{equation}
\begin{equation}
\label{eqn:constraint4}
\sum\limits_{i, j} A[i, j] = 2m,
\end{equation}
\begin{equation}
\label{eqn:constraint5}
\forall {(1\le i < n)}: 
\begin{pmatrix} 
\delta \le \sum\limits_{1\le j \le n}A[i, j] \le \Delta \\ 
\min\limits_i \left(\sum\limits_{1\le j \le n}A[i, j]\right) = \delta \\	
\max\limits_i \left(\sum\limits_{1\le j \le n}A[i, j]\right) = \Delta
\end{pmatrix}.
\end{equation}
Constraint \ref{eqn:constraint1} is referred to the symmetry of adjacency matrix and to the absence of loops. 
Constraints \ref{eqn:constraint2} and \ref{eqn:constraint3} stand for
no 3- and 4-cycles. Constraint \ref{eqn:constraint4} fixes the number of edges. 
Constraint \ref{eqn:constraint5} introduces degrees of vertices.
\par 
Similarly to ~\cite{Codish1} we alternatively introduce constraints \ref{eqn:constraint2a} 
and \ref{eqn:constraint3a} which generate $O(n^3)$ basic constraints, instead of $O(n^4)$ 
as constraints \ref{eqn:constraint2} and \ref{eqn:constraint3} do.
We introduce additional boolean variables $x_{i, k}$ and $x_{i, j, k}$:
$$\forall {i < k}: x_{i, j, k} \leftrightarrow A[i, j] \wedge A[j, k],$$
$$\forall {i < k}: x_{i, k} \leftrightarrow \bigvee \left\{ x_{i, j, k}\ |\ j \ne i, j \ne k \right\}.$$
With these variables we can express constraints \ref{eqn:constraint2} and \ref{eqn:constraint3}
with less amount of basic constraints:
\begin{equation}
\label{eqn:constraint2a}
\forall {i<k}: A[i, k] + x_{i, k} < 2 \tag{$2^\prime$},
\end{equation}
\begin{equation}
\label{eqn:constraint3a}
\forall {i<k}: \sum_j x_{i, j, k} < 2 \tag{$3^\prime$}.
\end{equation}

To optimize our model we used some properties of graphs from $EX(n; C_3, C_4)$,
proved in~\cite{Garnick2}:
\begin{equation}
\label{eqn:constraint6}
n \ge 1 + \Delta \delta \ge 1 + \delta^2,
\end{equation}
\begin{equation}
\label{eqn:constraint7}
\delta \ge m - ex(n - 1; C_3, C_4),
\end{equation}
\begin{equation}
\label{eqn:constraint8}
\Delta \ge \lceil 2m / n \rceil.
\end{equation}
	
In our experiments we use both versions of model (with constraints 
\ref{eqn:constraint2}-\ref{eqn:constraint3} and with 
\ref{eqn:constraint2a}-\ref{eqn:constraint3a}) and compare six configurations:
baseline (no breaking symmetries) and breaking symmetries with
$P_\bfs$, $P_\bfs^+$, $P_\bfs^*$, $P_\mathrm{UNAVOID}$ from~\cite{cuong2016computing} and $\mathrm{sb}_\ell^*$ from~\cite{Codish1}.
Computations were performed by AMD Opteron 6378 @ 2.4 GHz on 4 cores and the time limit was one hour. 
CSP model was compiled into SAT by BEE~\cite{metodi2012compiling}, SAT instance was solved
by treengeling~\cite{biere2016splatz}. We compare an efficiency of predicates on both satisfiable
($m = ex(n; C_3, C_4)$) and unsatisfiable ($m = ex(n; C_3, C_4) + 1$) cases. 
It turned out that constraints \ref{eqn:constraint2a}-\ref{eqn:constraint3a} 
are almost always more efficient than \ref{eqn:constraint2}-\ref{eqn:constraint3}, so 
only the former results are presented in Table \ref{table:table2} (sat case) and Table \ref{table:table2U} (unsat case). A ``--'' denotes that computations exceeded the time limit of four hours.
Some examples of the solutions found are shown in Figure~\ref{fig:solution_34_10}.

It worth to note that modern solvers use randomization a lot during the search. It implies that 
for the case of solution existence a time to find the solution may vary a lot. But this is not the
case for the non-existence of solution, because in this case the solver has to traverse the whole 
search space no matter in what order. In this work to get statistically valuable results we perform 
a series of 50 experiments for each combination of $n$ and a symmetry break in the sat case. For the
unsat case we perform only 5 experiments and it turned out that all results were the same. For the
sat case we present a median of 50 measured values.

\newcommand{\scalep}{0.02 * 0.8}

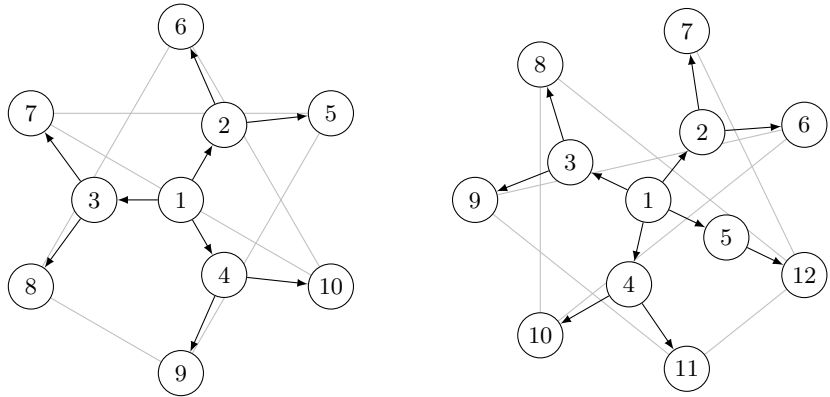
\begin{figure}
\centering
\begin{minipage}{0.5\textwidth}
\centering
\begin{tikzpicture}
	\tikzstyle{every node}=[draw,shape=circle,inner sep=0mm, minimum size=0.6cm,fill=white];
	\node (v1) at (\scalep * 142.71,162 * \scalep) {\small $1$};
	\node (v2) at (\scalep * 178.71,224.35 * \scalep) {\small $2$};
	\node (v3) at (\scalep * 70.708,162 * \scalep) {\small $3$};
	\node (v4) at (\scalep * 178.71,99.646 * \scalep) {\small $4$};
	\node (v5) at (\scalep * 267.42,234 * \scalep) {\small $5$};
	\node (v6) at (\scalep * 142.71,306 * \scalep) {\small $6$};
	\node (v7) at (\scalep * 18,234 * \scalep) {\small $7$};
	\node (v8) at (\scalep * 18,90 * \scalep) {\small $8$};
	\node (v9) at (\scalep * 142.71,18 * \scalep) {\small $9$};
	\node (v10) at (\scalep * 267.42,90 * \scalep) {\small $10$};	

    \begin{pgfonlayer}{background}
      \draw [\notinbfs]  (v5) -- (v7);
      \draw [\notinbfs]  (v5) -- (v9);
      \draw [\notinbfs]  (v6) -- (v8);
      \draw [\notinbfs]  (v6) -- (v10);
      \draw [\notinbfs]  (v7) -- (v10);
      \draw [\notinbfs]  (v8) -- (v9);
      \draw [arrows={-latex}] (v1) -- (v2);
      \draw [arrows={-latex}] (v1) -- (v3);
      \draw [arrows={-latex}] (v1) -- (v4);
      \draw [arrows={-latex}] (v2) -- (v5);
      \draw [arrows={-latex}] (v2) -- (v6);
      \draw [arrows={-latex}] (v3) -- (v7);
      \draw [arrows={-latex}] (v3) -- (v8);
      \draw [arrows={-latex}] (v4) -- (v9);
      \draw [arrows={-latex}] (v4) -- (v10);
    \end{pgfonlayer}
\end{tikzpicture}
\end{minipage}%
\begin{minipage}{0.5\textwidth}
\centering
\begin{tikzpicture}
	\tikzstyle{every node}=[draw,shape=circle,inner sep=0mm, minimum size=0.6cm,fill=white];
	\node (v1) at (\scalep * 162,159.89 * \scalep) {\small $1$};
	\node (v2) at (\scalep * 206.89,216.18 * \scalep) {\small $2$};
	\node (v3) at (\scalep * 97.13,191.13 * \scalep) {\small $3$};
	\node (v4) at (\scalep * 145.98,89.693 * \scalep) {\small $4$};
	\node (v5) at (\scalep * 226.87,128.65 * \scalep) {\small $5$};
	\node (v6) at (\scalep * 291.74,222.37 * \scalep) {\small $6$};
	\node (v7) at (\scalep * 194.04,300.28 * \scalep) {\small $7$};
	\node (v8) at (\scalep * 72.217,272.47 * \scalep) {\small $8$};
	\node (v9) at (\scalep * 18,159.89 * \scalep) {\small $9$};
	\node (v10) at (\scalep * 72.217,47.304 * \scalep) {\small $10$};
	\node (v11) at (\scalep * 194.04,19.498 * \scalep) {\small $11$};
	\node (v12) at (\scalep * 291.74,97.409 * \scalep) {\small $12$};
    
    \begin{pgfonlayer}{background}
      \draw [\notinbfs]  (v6) -- (v9);
      \draw [\notinbfs]  (v6) -- (v10);
      \draw [\notinbfs]  (v7) -- (v12);
      \draw [\notinbfs]  (v8) -- (v10);
      \draw [\notinbfs]  (v8) -- (v12);
      \draw [\notinbfs]  (v9) -- (v11);
      \draw [\notinbfs]  (v11) -- (v12);
      \draw [arrows={-latex}] (v1) -- (v2);
      \draw [arrows={-latex}] (v1) -- (v3);
      \draw [arrows={-latex}] (v1) -- (v4);
      \draw [arrows={-latex}] (v1) -- (v5);
      \draw [arrows={-latex}] (v2) -- (v6);
      \draw [arrows={-latex}] (v2) -- (v7);
      \draw [arrows={-latex}] (v3) -- (v8);
      \draw [arrows={-latex}] (v3) -- (v9);
      \draw [arrows={-latex}] (v4) -- (v10);
      \draw [arrows={-latex}] (v4) -- (v11);
      \draw [arrows={-latex}] (v5) -- (v12);

    \end{pgfonlayer}
\end{tikzpicture}
\end{minipage}
\caption{Graphs with 10 and 12 vertices without cycles of length 3 and 4 and with the maximum number of edges (arrows are the arcs of a BFS-tree, gray edges are the rest edges of a graph)}
\label{fig:solution_34_10}
\end{figure}

\begin{table}[ht]
\caption{Test SAT case with $m = ex(n; C_3, C_4)$ and constraints \ref{eqn:constraint2a} and \ref{eqn:constraint3a} ($O(n^3)$)}
\label{table:table2}

  \centering
  \setlength\tabcolsep{4pt}
  \begin{tabular}{lllllll}
  \hline\noalign{\smallskip}
$n$ & $\varnothing$, sec & $P_\bfs$, sec & $P_\bfs^+$, sec & $P_\bfs^*$, sec & $P_\mathrm{UNAVOID}$, sec & $\mathrm{sb}_\ell^*$, sec \\
  \noalign{\smallskip}\hline\noalign{\smallskip}
   18 &    5.68 &   \bf 1.04 &    1.08 &    4.12  &   10.38 &   14.37 \\
   19 &    0.98 &   \bf 0.92 &    1.08 &    1.35  &    1.37 &    1.01 \\
   20 &  126.10 &  344.95 &   77.81 & \bf 42.51  & 1234.16 & 2517.60 \\
   21 &  985.73 &  617.24 & \bf 115.89 &  239.00  & 2974.09 & 3528.72 \\
   22 & 2139.65 &  455.36 &  141.15 &  \bf 71.21  & --- & --- \\
   23 & ---     & ---     &   63.54 &  \bf 58.71  & --- & --- \\
   24 & 1765.81 &   52.30 &   \bf 6.22 &   26.64 & 1237.63 & --- \\
   25 & ---     & ---     & --- & \bf 2518.18 & --- & --- \\
   26 & ---     & ---     &  433.77 & \bf 161.68 & --- & --- \\
   27 & ---     & ---     & 1982.06 & \bf 349.10  & --- & --- \\
  \noalign{\smallskip}\hline
  \end{tabular}

\end{table}
\begin{table}[ht]
\caption{Test UNSAT case with $m = ex(n; C_3, C_4) + 1$ and constraints 
\ref{eqn:constraint2a} and \ref{eqn:constraint3a} ($O(n^3)$)}
\label{table:table2U}
\centering
  \setlength\tabcolsep{4pt}
  \begin{tabular}{lllllll}
  \hline\noalign{\smallskip}
  $n$ & $\varnothing$, sec & $P_\bfs$, sec & $P_\bfs^+$, sec & $P_\bfs^*$, sec & $P_\mathrm{UNAVOID}$, sec & $\mathrm{sb}_\ell^*$, sec \\
  \noalign{\smallskip}\hline\noalign{\smallskip}
   10 &   66.98 &    0.59 &   \bf 0.26 &    0.43 &    0.35 & \\
   11 & 2572.06 &    3.03 &    0.79 &   \bf 0.40 &    1.28 & \\
   12 & --- &    5.53 &    2.41 & \bf 1.05 &    5.01 & 2368.04 \\
   13 & --- &   17.79 &    5.76 & \bf 4.23 &   14.80 & --- \\
   14 & --- &   58.63 &   19.36 &   \bf 8.82 &   65.97 & --- \\
   15 & --- &  279.17 &   76.81 & \bf 21.82 &  254.83 & --- \\
   16 & --- & 2408.45 &  597.59 &  \bf 95.28 & 1455.55 & --- \\
   17 & --- & --- & --- & \bf 514.81 & --- & --- \\
   18 & --- & --- & --- & \bf 2773.63 & --- & --- \\
  \noalign{\smallskip}\hline
  \end{tabular}

\end{table}

\subsection{Determining of $ex(n; C_4)$}
\label{sec:4.2}

This problem was considered in~\cite{Clapham1}. We reduce a problem
into CSP instance by the same model as in the previous problem but without
constraints \ref{eqn:constraint2}, \ref{eqn:constraint6} and \ref{eqn:constraint7}.
We add new problem-specific constraints to the model (these constraints were studied in~\cite{Clapham1}):

\begin{equation}
\label{eqn:constraint9}
\delta \le \Delta,
\end{equation}

\begin{equation}
\label{eqn:constraint10}
\Delta(\delta - 1) \le n - 1,
\end{equation}

\begin{equation}
\label{eqn:constraint11}
\delta \le \frac{1}{2}\left(1 + \sqrt{4n - 3}\right).
\end{equation}

An asymptotically precise upper bound has been proven in ~\cite{Jukna1}: 
$$ex(n; C_4) \le \frac{n}{4}\left(1 + \sqrt{4n - 3}\right).$$

The experimental setup was the same as in the previous problem with the time limit of one hour. The results are presented in Table \ref{table:table4} (sat case) and Table \ref{table:table6} (unsat case). Values
in the tables are median of 50 and 5 measurements for the sat and unsat cases correspondingly.
 
\begin{table}[ht]
\centering
\setlength\tabcolsep{4pt}

\caption{Test case SAT with ${m = ex(n; C_4)}$ with constraints 
\ref{eqn:constraint3a} ($O(n^3)$)}
\label{table:table4}
\begin{tabular}{lllllll}
\hline\noalign{\smallskip}
$n$ & $\varnothing$, sec & $P_\bfs$, sec & $P_\bfs^+$, sec & $P_\bfs^*$, sec & $P_\mathrm{UNAVOID}$, sec & $\mathrm{sb}_\ell^*$, sec \\
\noalign{\smallskip}\hline\noalign{\smallskip}
   17 & \bf 42.50   &  297.69     &  280.39     &   57.72     &  205.02     &   55.24     \\
   18 & 1277.22     & 1143.39     & 1279.65     & \bf 258.81  & 1500.15     & 3425.28     \\
   19 & \bf 111.94  &  742.81     & 1445.58     &  292.67     & 3422.95     &  546.35     \\
   20 & ---         & ---         & ---         & \bf 1489.93 & ---         & ---         \\
   21 & ---         & ---         & ---         & \bf 2548.04 & ---         & ---         \\
   22 & ---         & ---         & ---         & \bf 2796.53 & ---         & ---         \\
\noalign{\smallskip}\hline
\end{tabular}

\end{table}
\begin{table}[ht]
\centering
\setlength\tabcolsep{4pt}

\caption{Test case UNSAT with ${m = ex(n; C_4) + 1}$ with constraints 
\ref{eqn:constraint3a} ($O(n^3)$)}
\label{table:table6}
\begin{tabular}{lllllll}
\hline\noalign{\smallskip}
$n$ & $\varnothing$, sec & $P_\bfs$, sec & $P_\bfs^+$, sec & $P_\bfs^*$, sec & $P_\mathrm{UNAVOID}$, sec & $\mathrm{sb}_\ell^*$, sec \\
\noalign{\smallskip}\hline\noalign{\smallskip}
   8 &    4.50     &    0.34     &    0.54     &    0.22     & \bf 0.21    &             \\
    9 &   49.93     &    0.95     &    0.53     & \bf 0.30     &    0.43     &             \\
   10 &  581.05     &    2.20     &    1.11     & \bf 0.51    &    0.92     &             \\
   11 & ---         &    7.56     &    2.84     & \bf   2.26     &    3.68     &             \\
   12 & ---         &   23.21     &   12.48     & \bf   6.29     &   12.73     & ---         \\
   13 &---          &   94.32     &   33.80     & \bf  12.95     &   43.07     & ---         \\
   14 &---          &  620.49     &  172.32     & \bf  30.45     &  233.13     & ---         \\
   15 &---          & 3285.48     &  768.11     & \bf  61.46     & 2008.36     &---          \\
   16 &---          & ---         & ---         & \bf 311.48     & ---         &---          \\
   17 &---          &---          & ---         & \bf 2181.37     & ---         &---          \\
\noalign{\smallskip}\hline
\end{tabular}

\end{table}

\section{Conclusion}
\label{sec:5}
We apply an approach from~\cite{ulyantsev2015bfs} to break symmetries in graph representations.
We also introduce and formally justify two improved predicates: 
$P_\bfs^+$ and $P_\bfs^*$. We demonstrate the efficiency of our approach on some problems from extremal graph theory and compare the impact with existing symmetry-breaking predicates.

\bibliographystyle{splncs03} 
\bibliography{list.bib}   

\begin{thebibliography}{10}
\providecommand{\url}[1]{\texttt{#1}}
\providecommand{\urlprefix}{URL }

\bibitem{Abajo1}
Abajo, E., Balbuena, C., Di{\'a}nez, A.: {New families of graphs without short
  cycles and large size}. Discrete applied mathematics  158(11),  1127--1135
  (2010)

\bibitem{Abajo2}
Abajo, E., Di{\'a}nez, A.: {Exact values of ex ($\nu$;$\{$C3, C4,\ldots,
  Cn$\}$)}. Discrete Applied Mathematics  158(17),  1869--1878 (2010)

\bibitem{biere2016splatz}
Biere, A.: {Splatz, Lingeling, Plingeling, Treengeling, YalSAT Entering the SAT
  Competition 2016}. SAT COMPETITION 2016 p.~44

\bibitem{Clapham1}
Clapham, C., Flockhart, A., Sheehan, J.: {Graphs without four-cycles}. Journal
  of Graph theory  13(1),  29--47 (1989)

\bibitem{Codish1}
Codish, M., Miller, A., Prosser, P., Stuckey, P.J.: {Breaking Symmetries in
  Graph Representation}. In: {IJCAI}. pp. 3--9 (2013)

\bibitem{crawford1994experiment}
Crawford, J.M., Baker, A.B.: {Experimental results on the application of
  Satisfiability Algorithms to Scheduling Problems}. In: Proceedings of the
  Twelfth National Conference on Artificial Intelligence ({AAAI}-94). vol.~2,
  pp. 1092--1097. {AAAI Press/MIT Press} (1994)

\bibitem{cuong2016computing}
Cuong, C., Heule, M.: {Computing Maximum Unavoidable Subgraphs Using SAT
  Solvers}. In: {International Conference on Theory and Applications of
  Satisfiability Testing}. pp. 196--211. Springer (2016)

\bibitem{devriendt2016improved}
Devriendt, J., Bogaerts, B., Bruynooghe, M., Denecker, M.: {Improved static
  symmetry breaking for SAT}. In: {International Conference on Theory and
  Applications of Satisfiability Testing}. pp. 104--122. Springer (2016)

\bibitem{Erdos1}
Erdos, P.: {Some recent progress on extremal problems in graph theory}. Congr.
  Numer  14,  3--14 (1975)

\bibitem{Garnick2}
Garnick, D.K., Kwong, Y., Lazebnik, F.: {Extremal graphs without three-cycles
  or four-cycles}. Journal of Graph Theory  17(5),  633--645 (1993)

\bibitem{Garnick1}
Garnick, D.K., Nieuwejaar, N.: {Non-isomorphic extremal graphs without
  three-cycles or four-cycles}. JCMCC  12,  33--56 (1992)

\bibitem{gent1999symmetry}
Gent, I.P., Smith, B.: {Symmetry breaking during search in constraint
  programming}. Citeseer (1999)

\bibitem{heule30quest}
Heule, M.J.: The quest for perfect and compact symmetry breaking for graph
  problems. In: Theory and Applications of Satisfiability Testing {--} SAT
  2016. pp. 228--245. Springer (2016)

\bibitem{itzhakov2016breaking}
Itzhakov, A., Codish, M.: {Breaking symmetries in graph search with canonizing
  sets}. Constraints  21(3),  357--374 (2016)

\bibitem{Jukna1}
Jukna, S.: {Extremal combinatorics: with applications in computer science}.
  Springer Science \& Business Media (2011)

\bibitem{lynce2007breaking}
Lynce, I., Marques-Silva, J.: {Breaking symmetries in SAT matrix models}. In:
  {International Conference on Theory and Applications of Satisfiability
  Testing}. pp. 22--27. Springer (2007)

\bibitem{metodi2012compiling}
Metodi, A., Codish, M.: {Compiling finite domain constraints to SAT with BEE}.
  Theory and Practice of Logic Programming  12(4-5),  465--483 (2012)

\bibitem{miller2012diamond}
Miller, A., Prosser, P.: {Diamond-free degree sequences}. arXiv preprint
  arXiv:1208.0460  (2012)

\bibitem{nethercote2007minizinc}
Nethercote, N., Stuckey, P.J., Becket, R., Brand, S., Duck, G.J., Tack, G.:
  {MiniZinc: Towards a standard CP modelling language}. In: {International
  Conference on Principles and Practice of Constraint Programming}. pp.
  529--543. Springer (2007)

\bibitem{stuckey2014minizinc}
Stuckey, P.J., Feydy, T., Schutt, A., Tack, G., Fischer, J.: {The MiniZinc
  challenge 2008--2013}. AI Magazine  35(2),  55--60 (2014)

\bibitem{ulyantsev2015bfs}
Ulyantsev, V., Zakirzyanov, I., Shalyto, A.: {BFS-based symmetry breaking
  predicates for DFA identification}. In: {International Conference on Language
  and Automata Theory and Applications}. pp. 611--622. Springer (2015)

\bibitem{Wang1}
Wang, P., Dueck, G.W., MacMillan, S.: {Using simulated annealing to construct
  extremal graphs}. Discrete Mathematics  235(1),  125--135 (2001)

\end{thebibliography}

\newpage

\appendix	
\section{Proof of theorems}
\label{apd:A}

\begin{theorem}
\label{thm:bfs}
$P_\bfs(G)$ is a symmetry-breaking predicate among connected graphs, i.e.
for each connected graph $G$ there exists a graph $G^\prime$ isomorphic to $G$ and
such that $P_\bfs(G^\prime) = 1$.
\end{theorem}

\begin{proof}
Let $v$ be an arbitrary vertex of $G$. Since $G$ is connected 
then BFS traversal that traverses all vertices of $G$ always exists.
Let $G^\prime$ be a graph obtained from $G$ by renumeration of vertices
in the order of mentioned BFS-traversal. $G^\prime$ is isomorphic to $G$
and is enumerated in BFS order. \qed
\end{proof}
\begin{theorem}
\label{thm:bfs_plus}
$P_\bfs^+(G)$ is a symmetry-breaking predicate among connected graphs, i.e.
for each connected graph $G$ there exists a graph $G^\prime$ isomorphic to $G$ and
such that $P_\bfs^+(G^\prime) = 1$.
\end{theorem}

\begin{proof}
The proof is almost equivalent to the proof of theorem \ref{thm:bfs}, but we should choose $v$ as a vertex with maximum degree. The rest of proof is the same.
\end{proof}
\begin{theorem}
\label{thm:bfs_star}
$P_\bfs^*(G)$ is a symmetry-breaking predicate among connected graphs, i.e.
for each connected graph $G$ there exists a graph $G^\prime$ isomorphic to $G$ and
such that $P_\bfs^*(G^\prime) = 1$.
\end{theorem}

\begin{lemma}
\label{lmm:lemma1}
Let $G$ be a connected graph and $P_\bfs(G)$ holds.
Let $\deg v_0 = k$ and $w_i = w(child_{G, i}(v_0))$.
Let $w_i \le w_{i + 1}$. Then $\exists G^\prime \simeq G:$
$P_\bfs(G^\prime) \wedge w_k = w^\prime_k (\forall k \neq i, 
\neq j) \wedge w^\prime_i \ge w^\prime_{i + 1}$, where 
$w^\prime_k = w(child_{G^\prime, k}(v_0^\prime))$.
\end{lemma}
\begin{proof}
We construct $G^\prime$ immediately by enumerating vertices from $G$.
Let $\left\{v_i\right\}_{i = 0}^{n - 1}$ be vertices of $G$ and 
$\left\{v^\prime_i\right\}_{i = 0}^{n - 1}$ be vertices of $G^\prime$.
Start vertex will rest unchanged ($v_0 = v^\prime_0$).
Vertices from the first layer (from $v_1$ to $v_k$) also
stay unchanged except of $v_i$ and $v_{i + 1}$, which are swapped
($v^\prime_{i + 1} = v_i, v^\prime_i = v_{i + 1}$).
Example of such renumbering is demonstrated in Figure \ref{fig:lemma3.1.1}.
We consider a BFS traversal of $G$ which visits 
vertices from $0$ to $k$ in a considered order: $v_0, v_1, \ldots, v_{i + 1}, v_i, \ldots, v_k$.
Rest of $v^\prime_i$ (for $i > k$) are defined as an index of $v_i$ in this BFS traversal.
$\{v^\prime_i\}_{i=0}^{n-1}$ is a sequence of vertices in BFS traversal so $P_\bfs(G^\prime)$ holds.
Also $G^\prime \simeq G$ since $V(G^\prime) = \left\{v^\prime_i\right\}_{i = 0}^{n - 1}$
is a permutation of $V(G) = \left\{v_i\right\}_{i = 0}^{n - 1}$.
Now we need to check if the weights of subtrees property holds.
\par
Throughout this proof we say that $S_1 \subset V(G_1)$ and 
$S_2 \subset V(G_2)$ are \emph{equal} (where $G_1$ and $G_2$
are isomorphic graphs and $G_1 = \pi(G_2)$) if $S_1 = \pi(S_2)$.

\textbf{Statement 1.} For all $j \ne i,\ j \ne i + 1$ 
holds $w(v^\prime_j) = w(v_j)$.
\par
Let $W_{i, i + 1}$ be a set of vertices that are descendants in BFS-tree
of $v_i$ or $v_{i + 1}$. Let $W^\prime_{i, i + 1}$ be a set of vertices 
that are descendants in BFS-tree of $v^\prime_i$ or $v^\prime_{i + 1}$.  
Now we check that these sets are \emph{equal}.
Let $W_j$ be an intersection of $V_{i, i + 1}$ and vertices from $j$-th 
layer in BFS-tree. Let $W^\prime_j$ be the same for $G^\prime$.
We show by induction that $\forall j: W_j = W^\prime_j$ (the equal sign means 
sets are \emph{equal}).
\\
\textbf{Basis}: $j = 1$, $W_1 = (v_i, v_{i + 1})$, $W^\prime_1 = 
(v^\prime_i, v^\prime_{i + 1})$ so $W_1 = W^\prime_1$.
\\
\textbf{Inductive step}: $W_j$ is a set of vertices for which
there is an edge from vertex from $W_{j - 1}$ and 
there is no edge from vertices with number less than numbers from $W_{j-1}$.
The same for $W^\prime_j$ in $G^\prime$. 
But following the induction hypothesis
$W_{j - 1} = W^\prime_{j - 1}$, sets of vertices
are reachable from them and not reachable from previous 
vertices are equal too.
\par
Consequently $\forall j: W_j = W^\prime_j$, so
$W_{i, i + 1} = W^\prime_{i, i + 1}$. It becomes clear,
that all differences in BFS-tree between $G$ and $G^\prime$ 
are entirely located in subtrees of $v_i$ and $v_{i + 1}$, so 
subtrees of other vertices remain unchanged and
$w(v^\prime_j) = w(v_j)$ for $j \ne i,\ j \ne i + 1$.

\textbf{Statement 2.} $w(v^\prime_i) \ge w(v_{i + 1})$.
\par
Consider an intersection of subtree of $v_{i + 1}$ and $j$-th 
layer of BFS-tree. Denote these sets in $G$ and $G^\prime$ as
$W_j$ and $W^\prime_j$, respectively. A vertex is in this set
if it wasn't visited earlier during BFS traversal and there exists
an edge from $W_{j - 1}$ to this vertex.
We show by induction that $\forall j: W_j \subseteq  
W^\prime_j$.
\\
\textbf{Basis}: $j = 1$, $W_1 = \{v_{i + 1}\}$, $W^\prime_1 = 
\{v^\prime_i\}$, so $W_1 = W^\prime_1$ and
$W_1 \subseteq W^\prime_1$.
\\
\textbf{Inductive step}: 
The order of vertex in $G^\prime$ was decremented ($v_{i+1} \rightarrow v_i^\prime$), 
so the order of all vertices in subtree in each layer is greater or equal
than the order of vertices in subtree of $v_{i+1}$.
The order of vertices relative to other subtrees remains unchanged.
Consequently, all vertices not visited by BFS traversal
remain not visited and there can be vertices that become
not visited after vertices swap. Following induction hypothesis,
$W_{j - 1} \subseteq W^\prime_{j - 1}$ and for each vertex 
set of new vertices at least as large as it was before swapping,
consequently $W_j \subseteq W^\prime_j$.

\textbf{Statement 3.} $w(v^\prime_{i + 1}) \le w(v_i)$.
\par
The proof is the same as for statement 2 with the only difference:
we should prove that $W_j \supseteq W^\prime_j$.
\par
From statements 2 and 3 it follows that 
$w^\prime_i \ge w_{i + 1} \ge w_i \ge w^\prime_{i + 1}$,
so $w^\prime_i \ge w^\prime_{i + 1}$ and the remaining weights 
remain unchanged (from statement 1). Also $G^\prime
\simeq G$ and $P_\bfs(G^\prime)$ holds. Hence the lemma holds.
\qed
\end{proof}

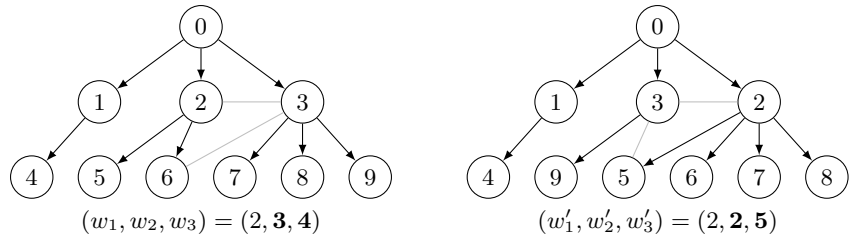
\begin{figure}[ht]
\centering
\begin{subfigure}
  \centering
    \begin{minipage}{0.48\textwidth}
    \centering
    \begin{tikzpicture}
	  \tikzstyle{every node}=[draw,shape=circle];
      \node (v0) at (2.65, 3) {\small $0$};
      \node (v1) at (1.30, 2) {\small $1$};
      \node (v2) at (2.65, 2) {\small $2$};
      \node (v3) at (4.00, 2) {\small $3$};
      \node (v4) at (0.4, 1) {\small $4$};
      \node (v5) at (1.3, 1) {\small $5$};
      \node (v6) at (2.2, 1) {\small $6$};
      \node (v7) at (3.1, 1) {\small $7$};
      \node (v8) at (4.0, 1) {\small $8$};
      \node (v9) at (4.9, 1) {\small $9$};
      \draw [arrows={-latex}] (v0) -- (v1);   
      \draw [arrows={-latex}] (v0) -- (v2);   
      \draw [arrows={-latex}] (v0) -- (v3);   
      \draw [arrows={-latex}] (v1) -- (v4);   
      \draw [arrows={-latex}] (v2) -- (v5);   
      \draw [arrows={-latex}] (v3) -- (v7);   
      \draw [arrows={-latex}] (v3) -- (v8);   	
      \draw [arrows={-latex}] (v3) -- (v9);   
      \draw [\notinbfs] (v2) -- (v3);
      \draw [\notinbfs] (v3) -- (v6);
      \draw [arrows={-latex}] (v2) -- (v6);
	\end{tikzpicture}\\
	$(w_1, w_2, w_3) = (2, \textbf 3, \textbf 4)$
    \end{minipage}
\end{subfigure}
\begin{subfigure}
  \centering
    \begin{minipage}{0.48\textwidth}
    \centering
    \begin{tikzpicture}
	  \tikzstyle{every node}=[draw,shape=circle];
      \node (v0) at (2.65, 3) {\small $0$};
      \node (v1) at (1.30, 2) {\small $1$};
      \node (v2) at (2.65, 2) {\small $3$};
      \node (v3) at (4.00, 2) {\small $2$};
      \node (v4) at (0.4, 1) {\small $4$};
      \node (v5) at (1.3, 1) {\small $9$};
      \node (v6) at (2.2, 1) {\small $5$};
      \node (v7) at (3.1, 1) {\small $6$};
      \node (v8) at (4.0, 1) {\small $7$};
      \node (v9) at (4.9, 1) {\small $8$};
      \draw [arrows={-latex}] (v0) -- (v1);   
      \draw [arrows={-latex}] (v0) -- (v2);   
      \draw [arrows={-latex}] (v0) -- (v3);   
      \draw [arrows={-latex}] (v1) -- (v4);   
      \draw [arrows={-latex}] (v2) -- (v5);   
      \draw [arrows={-latex}] (v3) -- (v7);   
      \draw [arrows={-latex}] (v3) -- (v8);   
      \draw [arrows={-latex}] (v3) -- (v9);   
      \draw [\notinbfs] (v2) -- (v3);
      \draw [\notinbfs] (v2) -- (v6);
      \draw [arrows={-latex}] (v3) -- (v6);
	\end{tikzpicture}\\
	$(w_1^\prime, w_2^\prime, w_3^\prime) = (2, \textbf 2, \textbf 5)$
    \end{minipage}
\end{subfigure}
\caption{BFS-tree before and after renumeration of vertices}
\label{fig:lemma3.1.1}
\end{figure}

\begin{lemma}
\label{lmm:lemma2}
Let $G$ be a connected graph and $P_\bfs(G)$ holds. Let $\deg v_0 = k$ and $w_i = w(child_i(v_0))$. Then $\exists\ G^\prime \simeq G: P_\bfs(G^\prime) 
\wedge {w_1^\prime \ge w_2^\prime \ge \ldots \ge w_k^\prime}$.
\end{lemma}
\begin{proof}
Consider a sequence $W^0_0 = \left[w_1, w_2, \ldots, w_k\right]$. 
Let $w_{i + 1}$ be a maximum element from $W_0^0$. 
Then $w_i \le w_{i + 1}$.
Following lemma \ref{lmm:lemma1}, there is a graph $G^1_1$ with 
a weights' sequence of subtrees of the root's children 
$W^1_0$ where $w^1_i \ge w^1_{i + 1}$.
Using this operation one can get a graph with a 
weights' sequence $W_1$ and $\forall j: w^1_1 \ge w^1_j$.
\par 
Similarly there exists a graph with weights' sequence $W_2$,
$w^2_1 \ge w^2_2$ and $\forall j \ge 2: w^2_2 \ge w^2_j$.
Continuing the process one can obtain a graph with a
weights' sequence $W_k$ where $w_1^k \ge w_2^k \ge \ldots \ge w_k^k$. 
\qed
\end{proof}

\begin{proof}[of theorem \ref{thm:bfs_star}]
By lemma \ref{lmm:lemma2} there exists a graph 
with non-ascending weights' sequence of root's children.
It is clear from the proof of lemma \ref{lmm:lemma1} that 
swapping of vertices have an influence only on vertices from their subtrees.
So one can apply lemma \ref{lmm:lemma2} to an
arbitrary vertex, not only to root.
Resulted graph preserves structure of BFS-tree
for vertices from subtrees of other vertices.
\par 
By lemma \ref{lmm:lemma2} there exists a graph 
$G^1_1$ with sorted weights' sequence
of root's children. One can apply this lemma
for the first vertex from the first layer and obtain
a graph $G^1_2$ with sorted weights' sequences of the root
and the first (leftmost) child of the root.
By applying this lemma in a such way one can obtain the
graph $G^\prime$ with the following property: weights' sequences 
of each vertex in graph are sorted in the non-ascending order. This ends
the proof of the theorem. \qed
\end{proof}

\end{document}